\input ./style/arxiv-vmsta.cfg
\documentclass[numbers,compress,v1.0.1]{vmsta}
\usepackage{vtexbibtags}
\usepackage{mathtools,mathbh}

\volume{5}
\issue{1}
\pubyear{2018}
\firstpage{81}
\lastpage{97}
\aid{VMSTA96}
\doi{10.15559/18-VMSTA96}
\articletype{research-article}

\startlocaldefs
\def\@journal@url{http://www.vmsta.org}
\def\@credit{%
  \vbox to 0pt{%
    \vskip-1.85pc
    \hskip-\textwidth
    \noindent
    \raise4mm\hbox to \textwidth{%
      \footnotesize
      \href{\@journal@url}{www.vmsta.org}%
      \hfill
      \href{\@vtex@url}{\includevtexlogo}%
      }%
    }%
  }

\newcommand{\rrvert}{\vert}
\newcommand{\llvert}{\vert}
\allowdisplaybreaks
\numberwithin{equation}{section}

\newtheorem{theorem}{Theorem}[section]
\newtheorem{lemma}[theorem]{Lemma}
\newtheorem{proposition}[theorem]{Proposition}
\newtheorem{corollary}[theorem]{Corollary}
\endlocaldefs

\begin{document}

\begin{frontmatter}
\pretitle{Research Article}

\title{Cliquet option pricing with Meixner processes}

\author{\inits{M.}\fnms{Markus}~\snm{Hess}\ead[label=e1]{Markus-Hess@gmx.net}}
\address{\institution{R+V Lebensversicherung AG},
Raiffeisenplatz 2, 65189 Wiesbaden, \cny{Germany}}




\markboth{M. Hess}{Cliquet option pricing with Meixner processes}

\begin{abstract}
We investigate the pricing of cliquet options in a
geometric Meixner model. The considered option is of monthly sum cap style
while the underlying stock price model is driven by a pure-jump
Meixner--L\'{e}vy process yielding Meixner distributed log-returns. In this
setting, we infer semi-analytic expressions for the cliquet option price by
using the probability distribution function of the driving Meixner--L\'{e}vy
process and by an application of Fourier transform techniques. In an
introductory section, we compile various facts on the Meixner distribution
and the related class of Meixner--L\'{e}vy processes. We also propose a
customized measure change preserving the Meixner distribution of any
Meixner process.
\end{abstract}
\begin{keywords}
\kwd{Cliquet option pricing}
\kwd{path-dependent exotic option}
\kwd{equity indexed annuity}
\kwd{log-return of financial asset}
\kwd{Meixner distribution}
\kwd{Meixner--L\'{e}vy process}
\kwd{stochastic differential equation}
\kwd{probability measure change}
\kwd{characteristic function}
\kwd{Fourier transform}
\end{keywords}
\begin{keywords}[MSC2010]%
\kwd[Primary ]{60G51}
\kwd{60H10}
\kwd{60H30}
\kwd[; Secondary ]{91B30}
\kwd{91B70}
\end{keywords}

\begin{keywords}[JEL]%
\kwd{G22}
\kwd{D52}

\end{keywords}

\received{\sday{27} \smonth{9} \syear{2017}}
\revised{\sday{2} \smonth{1} \syear{2018}}
\accepted{\sday{21} \smonth{1} \syear{2018}}
\publishedonline{\sday{12} \smonth{2} \syear{2018}}
\end{frontmatter}

\section{Introduction}\label{sec1}

Cliquet option based contracts constitute a customized subclass of equity
indexed annuities. The underlying options commonly are of monthly sum cap
style paying a credited yield based on the sum of monthly-capped rates
associated with some reference stock index. In this regard, cliquet type
investments belong to the class of path-dependent exotic options. In \cite{15}
cliquet options are regarded as ``the height of fashion in the world of
equity derivatives''. In the literature, there are different pricing
approaches for cliquet options involving e.g. partial differential
equations (see \cite{15}), Monte Carlo techniques (see \cite{2}), numerical recursive
algorithms related to inverse Laplace transforms (see \cite{9}) and analytical
computation methods (see \cite{3,7,8}). The present article belongs to the
last category.

The aim of the present paper is to provide analytical pricing formulas for
globally-floored locally-capped cliquet options with multiple resetting
times where the underlying reference stock index is driven by a pure-jump
time-homogeneous Meixner--L\'{e}vy process. In this setup, we derive cliquet
option price formulas under two different approaches: once by using the
distribution function of the driving Meixner--L\'{e}vy process and once by
applying Fourier transform techniques (as proposed in \cite{8}). All in all, the
present article can be seen as an accompanying (but to a large degree
self-contained) paper to \cite{8}, as it presents a specific application of the
results derived in \cite{8} to the class of Meixner--L\'{e}vy processes.

The paper is organized as follows: In Section~\ref{sec2} we compile facts on the
Meixner distribution and the related class of stochastic Meixner--L\'{e}vy
processes. In Section~\ref{sec3} we introduce a geometric pure-jump stock price
model driven by a Meixner--L\'{e}vy process. In Section~\ref{sec3.1} we establish a
customized structure preserving measure change from the risk-neutral to the
physical probability measure. Section~\ref{sec4} is dedicated to the pricing of
cliquet options. We obtain semi-analytic expressions for the cliquet option
price by using the probability distribution function of the driving
Meixner--L\'{e}vy process in Section~\ref{sec4.1} and by an application of Fourier
transform techniques in Section~\ref{sec4.2}. In Section~\ref{sec5} we draw the
conclusions.

\section{A review of Meixner processes}\label{sec2}

Let $ ( \Omega, \mathbb{F},  ( \mathcal{F}_{t}  )_{t\in  [
0, T  ]}, \mathbb{Q}  )$ be a filtered probability space
satisfying the usual hypotheses, i.e. $\mathcal{F}_{t} = \mathcal{F}_{t +}
\coloneqq \cap_{s > t} \mathcal{F}_{s}$ constitutes a right-continuous
filtration and $\mathbb{F}$ denotes the sigma-algebra augmented by all
$\mathbb{Q}$-null sets (cf. p. 3 in \cite{10}). Here, $\mathbb{Q}$ is a
risk-neutral probability measure and $0< T < \infty$ denotes a finite time
horizon. In the following, we compile various facts on the Meixner
distribution and Meixner--L\'{e}vy processes from \cite{1,6,12,13} and
\cite{14}.

A real-valued, c\`{a}dl\`{a}g, pure-jump, time-homogeneous L\'{e}vy process
$M =\break  ( M_{t}  )_{t\in  [ 0, T  ]}$ (with independent and
stationary increments) satisfying $M_{0} =0$ is called Meixner (-L\'{e}vy)
process with scaling parameter $\alpha >0$, shape/skewness parameter
$\beta\in  ( -\pi, \pi  )$, peakedness parameter $\delta >0$ and
location parameter $\mu \in\mathbb{R}$, if $M_{t}$ possesses the
L\'{e}vy--It\^{o} decomposition
\begin{equation}
M_{t} = \theta t + \int_{0}^{t} \int
_{\mathbb{R}_{0}} z d \tilde{N}^{\mathbb{Q}} ( s, z )\label{eq2.1}
\end{equation}
where $\mathbb{R}_{0}\coloneqq\mathbb{R}\setminus  \{ 0  \}$,
the drift parameter
\begin{equation}
\theta\coloneqq\mu + \delta\alpha \tan ( {\beta} / {2} )\label{eq2.2}
\end{equation}
is a real-valued constant and the $\mathbb{Q}$-compensated Poisson random
measure (PRM) is given by
\begin{equation}
d \tilde{N}^{\mathbb{Q}} ( s, z ) \coloneqq dN ( s, z ) -d \nu^{\mathbb{Q}} (
z ) ds\label{eq2.3}
\end{equation}
with positive and finite Meixner-type L\'{e}vy measure
\begin{equation}
d \nu^{\mathbb{Q}} ( z ) \coloneqq\delta \frac{e^{{\beta z} /
{\alpha}}} {z \sinh  ( {\pi z} / {\alpha}  )} dz\label{eq2.4}
\end{equation}
(cf. \cite{12,13}, Eq. (3) in \cite{6}) satisfying $\nu^{\mathbb{Q}}  (
 \{ 0  \}  ) =0$ and
\[
\int_{\mathbb{R}_{0}} \bigl( 1 \wedge z^{2} \bigr) d
\nu^{\mathbb{Q}} ( z ) < \infty.
\]
We denote the L\'{e}vy triplet of $M_{t}$ by $ ( \theta,0,
\nu^{\mathbb{Q}}  )$. (Note that this notation is not entirely
consistent with \cite{6,8}.) We recall that $M_{t}$ possesses moments of all
orders (cf. Section~5.3.10 in \cite{13}). Evidently, $M_{t}$ has no Brownian
motion part. Since
\[
\int_{\mathbb{R}_{0}} \llvert z \rrvert d \nu^{\mathbb{Q}} ( z ) =
\infty
\]
the process $M_{t}$ possesses infinite variation (cf. Section~5.3.10 in
\cite{13}). We write for any fixed $t\in  [ 0, T  ]$
\[
M_{t} \sim\mathcal{M} ( \alpha, \beta, \delta t, \mu t )
\]
(cf. Section~3.6 in \cite{1}) and say that $M$ is Meixner distributed under
$\mathbb{Q}$ with parameters $\alpha$, $\beta$, $\delta$ and $\mu$. From
(\ref{eq2.1}) and (\ref{eq2.2}) we instantly receive the mean value
\begin{equation}
\mathbb{E}_{\mathbb{Q}} [ M_{t} ] = \theta t = \mu t + \delta t
\alpha \tan ( {\beta} / {2} )\label{eq2.5}
\end{equation}
standing in accordance with Eq. (11) in \cite{6}. The variance, skewness and
kurtosis of $M_{t}$ are respectively given by
\begin{align*}
\mathbh{Var}_{\mathbb{Q}} [ M_{t} ] &= \frac{\delta t} {2}
\frac{\alpha^{2}} {\cos^{2}  ( {\beta} / {2}  )},\qquad \mathbb{S}_{\mathbb{Q}} [ M_{t} ] = \sqrt{{2} / { (
\delta t )}} \sin ( {\beta} / {2} ),\\
 \mathbb{K}_{\mathbb{Q}} [ M_{t} ]
&=3+ \frac{2 - \cos  ( \beta  )} {\delta t}
\end{align*}
(cf. Table~6 in \cite{1}). Furthermore, for all $x \in\mathbb{R}$ and $t\in
 [ 0, T  ]$ the real-valued probability density function (pdf) of
$M_{t}$ under $\mathbb{Q}$ reads as
\begin{equation}
f_{M_{t}} ( x ) \coloneqq \frac{ ( 2 \cos  ( {\beta} /
{2}  )  )^{2 \delta t}} {2 \pi \alpha \varGamma  ( 2 \delta t
 )} e^{{\beta  ( x-\mu t  )} / {\alpha}} \biggl\llvert
\varGamma \biggl( \delta t + i \frac{x-\mu t} {\alpha} \biggr) \biggr\rrvert
^{2}\label{eq2.6}
\end{equation}
(cf. \cite{1,12}, Eq. (4) in \cite{6}) wherein
\[
\varGamma ( \zeta ) \coloneqq \int_{0}^{\infty}
u^{\zeta- 1} e^{-u} du
\]
denotes the gamma function which is defined for all $\zeta \in\mathbb{C}$
with $R e  ( \zeta  ) >0$. Taking the definition of the gamma
function and Euler's formula into account, we get
\begin{align*}
\varGamma \biggl( \delta t + i \frac{x-\mu t} {\alpha} \biggr) &= \int
_{0^{
+}}^{\infty} u^{\delta t- 1} e^{-u} \cos
\biggl( \frac{x-\mu t} {\alpha} \ln u \biggr) du \\
&\quad+ i \int_{0^{ +}}^{\infty}
u^{\delta t- 1} e^{-u} \sin \biggl( \frac{x-\mu t} {\alpha} \ln u \biggr) du
\end{align*}
which implies
\begin{align*}
\biggl\llvert \varGamma \biggl( \delta t + i \frac{x-\mu t} {\alpha} \biggr) \biggr
\rrvert ^{2}
&= \Biggl( \int_{0^{ +}}^{\infty} u^{\delta t- 1}
e^{-u} \cos \biggl( \frac{x-\mu t} {\alpha} \ln u \biggr) du
\Biggr)^{2}\\
&\quad + \Biggl( \int_{0^{
+}}^{\infty}
u^{\delta t- 1} e^{-u} \sin \biggl( \frac{x-\mu t} {\alpha} \ln u \biggr) du
\Biggr)^{2}.
\end{align*}
Note that the latter object appears in (\ref{eq2.6}). The cumulative distribution
function (cdf) of $M_{t}$ does not possess a closed form representation but
it can be computed numerically. Further on, the characteristic function of
$M_{t}$ can be computed by the L\'{e}vy--Khinchin formula (see e.g. \cite{4,5,11,13}) due to
\begin{equation}
\phi_{M_{t}} ( u ) \coloneqq \mathbb{E}_{\mathbb{Q}} \bigl[
e^{iu M_{t}} \bigr] = e^{\psi  ( u  ) t}\label{eq2.7}
\end{equation}
with $i^{2} = - 1$, $u \in\mathbb{R}$, $t\in  [ 0, T  ]$ and a
characteristic exponent
\begin{equation}
\psi ( u ) \coloneqq iu \biggl[ \mu + \delta\alpha \tan \biggl(
\frac{\beta} {2} \biggr) \biggr] + \delta \int_{\mathbb{R}_{0}}
\frac{e^{iuz} - 1 - iuz} {z} \frac{e^{{\beta z} / {\alpha}}} {\sinh  (
{\pi z} / {\alpha}  )} dz.\label{eq2.8}
\end{equation}
Moreover, let us define the Fourier transform, respectively inverse Fourier
transform, of a deterministic function $q\in \mathcal{L}^{1}  (
\mathbb{R}  )$ via
\[
\hat{q} ( y ) \coloneqq \int_{\mathbb{R}} q ( x ) e^{iyx}
dx, q ( x ) = \frac{1} {2 \pi} \int_{\mathbb{R}} \hat{q} ( y )
e^{-iyx} dy.
\]
Then for all $u \in\mathbb{R}$ and $t\in  [ 0, T  ]$ we receive
the well-known relationship
\[
\phi_{M_{t}} ( u ) = \hat{f}_{M_{t}} ( u )
\]
where $\hat{f}_{M_{t}}$ denotes the Fourier transform of the density
function $f_{M_{t}}$ defined in (\ref{eq2.6}). An application of the inverse
Fourier transform yields
\[
f_{M_{t}} ( x ) = \frac{1} {2 \pi} \int_{\mathbb{R}}
e^{\psi
 ( u  ) t-iux} du
\]
thanks to (\ref{eq2.7}). On the other hand, from Eq. (1) in \cite{6} we know that
\begin{equation}
\phi_{M_{t}} ( u ) = e^{iu\mu t} \biggl( \frac{\cos  (
{\beta} / {2}  )} {\cosh  ( { ( \alpha u-i\beta  )} / {2}
 )}
\biggr)^{2 \delta t}\label{eq2.9}
\end{equation}
where $u \in\mathbb{R}$ and $t\in  [ 0, T  ]$. Taking the
logarithm in (\ref{eq2.7}) and (\ref{eq2.9}), we finally deduce
\begin{equation}
\psi ( u ) = iu\mu +2 \delta \biggl[ \ln \cos \biggl( \frac{\beta}{
2} \biggr)
- \ln \cosh \biggl( \frac{\alpha u-i\beta} {2} \biggr) \biggr].\label{eq2.10}
\end{equation}
Further on, for the Meixner distribution the following properties are
well-known (cf. \cite{14}, Section~5.3.10 in \cite{13}, Section~3.6 in \cite{1}, Corollary
1 in \cite{6}).

\begin{lemma}\label{lem2.1}
\begin{enumerate}
\item[(a)] If $X \sim\mathcal{M}  ( \alpha, \beta,
\delta, \mu  )$, then $c X+m \sim\mathcal{M}  ( c\alpha,
\beta, \delta, c\mu + m  )$ with constants $c >0$
and $m \in\mathbb{R}$.
\item[(b)] If $X_{1} \sim\mathcal{ M}  ( \alpha, \beta, \delta_{1}, \mu_{1}
 )$ and $X_{2} \sim\mathcal{ M}  ( \alpha, \beta,
\delta_{2}, \mu_{2}  )$ are independent random variables,
then $X_{1} + X_{2}\sim \mathcal{M}  ( \alpha, \beta, \delta_{1} +
\delta_{2}, \mu_{1} + \mu_{2}  )$.
\item[(c)] The characteristic function $\phi_{X}  ( u; \alpha,
\beta, \delta, \mu  )$ of a Meixner distributed random
variable $X\sim \mathcal{M}  ( \alpha, \beta, \delta, \mu  )$
satisfies
\[
\phi_{X} ( u; \alpha, \beta, \delta, \mu ) = \phi_{X} ( u;
\alpha, \beta, {\delta} / {n}, {\mu} / {n} )^{n}
\]
for arbitrary $n\in \mathbb{N}$ such that the Meixner
distribution is infinitely divisible.
\end{enumerate}
\end{lemma}

\section{A stock price model driven by a Meixner process}\label{sec3}

Let $t\in  [ 0, T  ]$ and define the stochastic stock price
process $S_{t}$ via
\begin{equation}
S_{t} \coloneqq S_{0} e^{M_{t} +bt}\label{eq3.1}
\end{equation}
with deterministic initial value $S_{0}$, a constant $b \in\mathbb{R}$ and
a real-valued Meixner process
\[
M_{t} = \theta t + \int_{0}^{t} \int
_{\mathbb{R}_{0}} z d \tilde{N}^{\mathbb{Q}} ( s, z )
\]
such as introduced in (\ref{eq2.1})--(\ref{eq2.4}). Here, the constant $b$ provides some
additional degree of freedom which is introduced in order to ensure the
arbitrage-freeness of the stock price model. More details on this topic
will be given below. Verify that (\ref{eq3.1}) belongs to the same model class
(geometric L\'{e}vy models) as (2.2)--(2.3) in \cite{8}. We next introduce the
historical filtration
\[
\mathcal{F}_{t} \coloneqq\sigma \{ S_{u}:0 \leq u\leq t \}
= \sigma \{ M_{u}:0 \leq u\leq t \}.
\]
Using It\^{o}'s formula, we obtain the stochastic differential equation
(SDE)
\[
\frac{d S_{t}} {S_{t-}} = \biggl( \theta + b + \int_{\mathbb{R}_{0}} \bigl[
e^{z} - 1 - z \bigr] d \nu^{\mathbb{Q}} ( z ) \biggr) dt + \int
_{\mathbb{R}_{0}} \bigl[ e^{z} - 1 \bigr] d
\tilde{N}^{\mathbb{Q}} ( t, z )
\]
under $\mathbb{Q}$. Let us further define the discounted stock price via
\[
\hat{S}_{t} \coloneqq \frac{S_{t}} {B_{t}}
\]
where $S_{t}$ is such as defined in (\ref{eq3.1}) and $B_{t} \coloneqq e^{rt}$ is
the value of a bank account with normalized initial capital $B_{0} =1$ and
risk-less interest rate $r>0$. Due to (\ref{eq3.1}) we find
\[
\hat{S}_{t} = S_{0} e^{M_{t} +  ( b - r  ) t}
\]
while It\^{o}'s formula yields the following SDE under $\mathbb{Q}$
\[
\frac{d \hat{S}_{t}} {\hat{S}_{t-}} = \biggl( \theta + b-r + \int_{\mathbb{R}_{0}} \bigl[
e^{z} - 1 - z \bigr] d \nu^{\mathbb{Q}} ( z ) \biggr) dt + \int
_{\mathbb{R}_{0}} \bigl[ e^{z} - 1 \bigr] d
\tilde{N}^{\mathbb{Q}} ( t, z ).
\]
In accordance to no-arbitrage theory, the discounted stock price $\hat{S}$
must form a martingale under the risk-neutral probability measure
$\mathbb{Q}$. For this reason, we require the drift restriction
\begin{equation}
b =r -\theta- \int_{\mathbb{R}_{0}} \bigl[ e^{z} - 1 - z
\bigr] d \nu^{\mathbb{Q}} ( z ).\label{eq3.2}
\end{equation}
With this particular choice of the coefficient $b$, we deduce
\[
\frac{d S_{t}} {S_{t-}} = r dt + \int_{\mathbb{R}_{0}} \bigl[
e^{z} - 1 \bigr] d \tilde{N}^{\mathbb{Q}} ( t, z )
\]
under $\mathbb{Q}$. Combining (\ref{eq3.1}) and (\ref{eq3.2}), we receive
\[
S_{t} = S_{0} e^{rt} \exp \Biggl\{ \int
_{0}^{t} \int_{\mathbb{R}_{0}} z d
\tilde{N}^{\mathbb{Q}} ( s, z ) - \int_{0}^{t} \int
_{\mathbb{R}_{0}} \bigl[ e^{z} - 1 - z \bigr] d
\nu^{\mathbb{Q}} ( z ) ds \Biggr\}
\]
where the last factor on the right hand side constitutes a Dol\'{e}ans-Dade
exponential which again shows the $\mathbb{Q}$-martingale property of the
discounted stock price process $\hat{S}_{t} = S_{t} e^{- rt}$. Moreover,
taking (\ref{eq2.2}), (\ref{eq2.4}), (\ref{eq2.8}) and (\ref{eq2.10}) into account, Eq.~(\ref{eq3.2}) can be
expressed as
\begin{equation}
b=r - \psi ( - i ) =r - \mu - 2\delta \ln \biggl( \frac{\cos
 ( {\beta} / {2}  )} {\cos  ( { ( \alpha+\beta  )} /
{2}  )} \biggr).\label{eq3.3}
\end{equation}
Unless otherwise stated, from now on we assume that the constant $b \in
\mathbb{R}$ appearing in (\ref{eq3.1}) is such as given in (\ref{eq3.3}). Though
constituting an admissible choice, taking $b=0$ in (\ref{eq3.1})--(\ref{eq3.3}) might be too
restrictive in practical applications. In the following, we investigate the
log-returns related to our model (\ref{eq3.1}). For an arbitrary time step $\varDelta
>0$ and $t \leq T - \varDelta$ we obtain
\[
\ln \biggl( \frac{S_{t+ \varDelta}} {S_{t}} \biggr) \cong M_{\varDelta} +b \varDelta \sim
\mathcal{M} \bigl( \alpha,\beta,\delta \varDelta, ( \mu+b ) \varDelta \bigr)
\]
by Lemma~\ref{lem2.1} (a). Here, the symbol $\cong$ denotes equality in
distribution. Hence, in our stock price model (\ref{eq3.1}) the log-returns are
Meixner distributed. We stress that in \cite{12} it was shown that the Meixner
distribution fits empirical financial log-returns very well. Furthermore,
for $n \in\mathbb{N}$ we introduce the time partition $\mathcal{P}
\coloneqq  \{ 0< t_{0} < t_{1} < \cdots < t_{n} \leq T  \}$ and
define the return/revenue process associated with the period $ [ t_{k -
1}, t_{k}  ]$ via
\[
R_{k} \coloneqq \frac{S_{t_{k}} - S_{t_{k - 1}}} {S_{t_{k - 1}}}
\]
where $k \in  \{ 1, \dots,n  \}$. A substitution of (\ref{eq3.1}) into the
latter equation yields
\begin{equation}
R_{k} = e^{Y_{t_{k}} - Y_{t_{k - 1}}} - 1\label{eq3.4}
\end{equation}
where
\[
Y_{t} \coloneqq M_{t} +bt \sim\mathcal{ M} \bigl( \alpha,
\beta,\delta t, ( \mu+b ) t \bigr)
\]
is a Meixner--L\'{e}vy process. Taking (\ref{eq2.1})--(\ref{eq2.4}) and (\ref{eq3.3}) into account,
we get
\begin{equation}
Y_{t} =\gamma t+ \int_{0}^{t} \int
_{\mathbb{R}_{0}} z dN ( s,z )\label{eq3.5}
\end{equation}
where
\[
\gamma \coloneqq r+\delta \biggl[ \alpha \tan \biggl( \frac{\beta} {2} \biggr)
- 2 \ln \biggl( \frac{\cos  ( {\beta} / {2}  )} {\cos
 ( { ( \alpha+\beta  )} / {2}  )} \biggr) - \int_{\mathbb{R}_{0}}
\frac{e^{{\beta z} / {\alpha}}} {\sinh  ( {\pi z}
/ {\alpha}  )} dz \biggr]
\]
is a real-valued constant. Recall that the Meixner--L\'{e}vy process $Y$
given in (\ref{eq3.5}) above just is a special case of the more general L\'{e}vy
process $X$ defined in Eq. (2.3) in \cite{8}. For this reason, the cliquet
option pricing results derived in \cite{8} simultaneously apply in our current
Meixner modeling case. More details on this topic are given in Section~\ref{sec4}
below. Also note that $R_{1}, \dots, R_{n}$ are $\mathbb{Q}$-independent
random variables and that $R_{k} > - 1$ $\mathbb{Q}$-almost surely for all
$k$. Since $Y$ is a L\'{e}vy process under $\mathbb{Q}$, we observe
$Y_{t_{k}} - Y_{t_{k-1}} \cong Y_{\tau}$ (stationary increments) where
$\tau\coloneqq t_{k} - t_{k-1}$ (equidistant partition). Here, the symbol
$\cong$ denotes equality in distribution. For the sake of notational
simplicity, we always work under the assumption of equidistant time points
in the following, unless otherwise stated. Taking (\ref{eq3.4}) into account, we
obtain the subsequent relationship between the cumulative distribution
functions of $R_{k}$ and $Y_{\tau}$
\begin{equation}
\mathbb{Q} ( R_{k} \leq\xi ) \mathbb{=Q} \bigl( Y_{\tau} \leq
\ln ( 1+\xi ) \bigr)\label{eq3.6}
\end{equation}
where $\xi>-1$ is an arbitrary real-valued constant.

\subsection{A structure preserving measure change to the physical
probability measure}\label{sec3.1}

Recall that we worked under the risk-neutral probability measure
$\mathbb{Q}$ in the previous sections. Since log-returns of financial
assets are commonly observed under the physical measure $\mathbb{P}$
(instead of under $\mathbb{Q}$), we establish a measure change from
$\mathbb{Q}$ to $\mathbb{P}$ in the sequel. In this context, we have to pay
special attention to the so-called structure preserving property of the
measure change, as the log-returns under $\mathbb{P}$ shall again follow a
Meixner distribution. In other words, the Meixner process $M_{t}$
introduced in (\ref{eq2.1}) under $\mathbb{Q}$ shall also be a Meixner process
under $\mathbb{P}$. First of all, for $t\in  [ 0,T  ]$ we define
the Radon--Nikodym density process
\[
\varLambda_{t} \coloneqq  \frac{d \mathbb{P}} {d \mathbb{Q}} \bigg\rrvert
_{\mathcal{F}_{t}} \coloneqq \exp \Biggl\{ \int_{0}^{t}
\int_{\mathbb{R}_{0}} h ( z ) d \tilde{N}^{\mathbb{Q}} ( s,z ) - \int
_{0}^{t} \int_{\mathbb{R}_{0}} \bigl[
e^{h  ( z
 )} -1- h ( z ) \bigr] d \nu^{\mathbb{Q}} ( z ) ds \Biggr\}
\]
where the $\mathbb{Q}$-compensated PRM $\tilde{N}^{\mathbb{Q}}$ and the
corresponding L\'{e}vy measure $\nu^{\mathbb{Q}}$ are such as defined in
(\ref{eq2.3}), respectively (\ref{eq2.4}), while $h  ( z  )$ is a time-independent
deterministic function on $\mathbb{R}_{0}$. Recall that we may write
\[
\varLambda_{t} = \frac{e^{L_{t}}} {\mathbb{E}_{\mathbb{Q}}  [ e^{L_{t}}
 ]}
\]
with a local $\mathbb{Q}$-martingale process
\[
L_{t} \coloneqq \int_{0}^{t} \int
_{\mathbb{R}_{0}} h ( z ) d \tilde{N}^{\mathbb{Q}} ( s,z )
\]
such that the density process $\varLambda$ is detected to be of Esscher
transform type. Note that $\varLambda$ is a discontinuous Dol\'{e}ans-Dade
exponential which constitutes a local martingale under $\mathbb{Q}$
satisfying the SDE\vadjust{\goodbreak}
\[
d \varLambda_{t} = \varLambda_{t-} \int_{\mathbb{R}_{0}}
\bigl[ e^{h  ( z
 )} -1 \bigr] d \tilde{N}^{\mathbb{Q}} ( t,z ).
\]
In accordance to Theorem 12.21 in \cite{5}, we further impose the Novikov
condition
\[
\mathbb{E}_{\mathbb{Q}} \Biggl[ \exp \Biggl\{ \int_{0}^{t}
\int_{\mathbb{R}_{0}} \bigl[ 1- e^{h  ( z  )} +h ( z )
e^{h  ( z  )} \bigr] d \nu^{\mathbb{Q}} ( z ) ds \Biggr\} \Biggr] <\infty
\]
for all $t\in  [ 0,T  ]$. Then it holds $\mathbb{E}_{\mathbb{Q}}
 [ \varLambda_{t}  ] \equiv1$ for all $t\in  [ 0,T  ]$ such
that $\varLambda$ constitutes a true $\mathbb{Q}$-martingale. Hence, we may
apply Girsanov's theorem stating that
\begin{equation}
d \tilde{N}^{\mathbb{P}} ( s,z ) \coloneqq dN ( s,z ) -d \nu^{\mathbb{P}} (
z ) ds\label{eq3.7}
\end{equation}
constitutes the $\mathbb{P}$-compensated Poisson random measure with
L\'{e}vy measure
\begin{equation}
d \nu^{\mathbb{P}} ( z ) \coloneqq e^{h  ( z  )} d \nu^{\mathbb{Q}} ( z ).\label{eq3.8}
\end{equation}
Note that the Novikov condition is equivalent to requiring that
\[
\int_{\mathbb{R}_{0}} \bigl[ 1- e^{h  ( z  )} +h ( z )
e^{h  ( z  )} \bigr] d \nu^{\mathbb{Q}} ( z ) <\infty
\]
since $h$ and $\nu^{\mathbb{Q}}$ both are deterministic. A combination of
(\ref{eq2.1}), (\ref{eq2.3}), (\ref{eq3.7}) and (\ref{eq3.8}) yields the following L\'{e}vy--It\^{o}
decomposition
\begin{equation}
M_{t} = \biggl( \theta+ \int_{\mathbb{R}_{0}} z \bigl[
e^{h  ( z
 )} -1 \bigr] d \nu^{\mathbb{Q}} ( z ) \biggr) t+ \int
_{0}^{t} \int_{\mathbb{R}_{0}} z d
\tilde{N}^{\mathbb{P}} ( s,z )\label{eq3.9}
\end{equation}
under $\mathbb{P}$ where $\theta$, $\nu^{\mathbb{Q}}$ and
$\tilde{N}^{\mathbb{P}}$ are such as defined in (\ref{eq2.2}), (\ref{eq2.4}) and (\ref{eq3.7}),
respectively. The remaining challenge now consists in finding an
appropriate function $h  ( z  )$ which, firstly, fulfills the
Novikov condition, secondly, guarantees that $\nu^{\mathbb{P}}$ in (\ref{eq3.8})
constitutes a L\'{e}vy measure of Meixner-type and, thirdly, ensures that
$M_{t}$ in (\ref{eq3.9}) is a Meixner--L\'{e}vy process. In this regard, we propose
to work with the specification
\begin{equation}
h ( z ) \coloneqq \frac{\beta^{*} -\beta} {\alpha} z\label{eq3.10}
\end{equation}
from now on. Herein, the constant skewness parameter $\beta^{*} \in  (
-\pi,\pi  )$ satisfies $\beta^{*} \neq\beta$ while $\alpha>0$ is the
scaling parameter introduced above. Note that taking $\beta^{*} =\beta$
would imply $h  ( z  ) \equiv0$ and hence, $\mathbb{P=Q}$.
Combining (\ref{eq3.8}) with (\ref{eq2.4}) and (\ref{eq3.10}), we deduce
\begin{equation}
d \nu^{\mathbb{P}} ( z ) =\delta \frac{e^{{\beta^{*} z} /
{\alpha}}} {z \sinh  ( {\pi z} / {\alpha}  )} dz\label{eq3.11}
\end{equation}
which constitutes a Meixner-type L\'{e}vy measure with parameters $\alpha$,
$\beta^{*}$ and $\delta$ [recall Eq.~(\ref{eq2.4})]. Moreover, with respect to
(\ref{eq3.8}) and (\ref{eq3.10}), we obtain
\[
\int_{\mathbb{R}_{0}} \bigl[ 1- e^{h  ( z  )} +h ( z )
e^{h  ( z  )} \bigr] d \nu^{\mathbb{Q}} ( z ) = \nu^{\mathbb{Q}} (
\mathbb{R}_{0} ) - \nu^{\mathbb{P}} ( \mathbb{R}_{0} ) +
\frac{\beta^{*} -\beta} {\alpha} \int_{\mathbb{R}_{0}} z d \nu^{\mathbb{P}} ( z )
\]
which is finite, because the L\'{e}vy measures $\nu^{\mathbb{Q}}$ and
$\nu^{\mathbb{P}}$ are finite. Thus, the function $h  ( z  )$
defined in (\ref{eq3.10}) indeed fulfills the Novikov condition. Further on, we
take (\ref{eq2.2}), (\ref{eq2.4}), (\ref{eq3.9}) and (\ref{eq3.10}) into account and receive
\begin{equation}
\mathbb{E}_{\mathbb{P}} [ M_{t} ] = \theta^{*} t\label{eq3.12}
\end{equation}
with drift parameter
\[
\theta^{*} \coloneqq \mu^{*} +\delta\alpha \tan \biggl(
\frac{\beta^{*}}{
2} \biggr)
\]
and a constant and real-valued location parameter
\begin{equation}
\mu^{*} \coloneqq\mu+\delta \biggl[ \alpha \tan \biggl(
\frac{\beta} {2} \biggr) -\alpha \tan \biggl( \frac{\beta^{*}} {2} \biggr) + \int
_{\mathbb{R}_{0}} \frac{e^{{\beta^{*} z} / {\alpha}} - e^{{\beta z} /
{\alpha}}} {\sinh  ( {\pi z} / {\alpha}  )} dz \biggr].\label{eq3.13}
\end{equation}
Note in passing that (\ref{eq3.12}) possesses the same structure as (\ref{eq2.5}). Also
verify that
\begin{equation}
\theta^{*} -\theta=\delta \int_{\mathbb{R}_{0}}
\frac{e^{{\beta^{*} z} /
{\alpha}} - e^{{\beta z} / {\alpha}}} {\sinh  ( {\pi z} / {\alpha}
 )} dz = \int_{\mathbb{R}_{0}} z \bigl[ d
\nu^{\mathbb{P}} ( z ) -d \nu^{\mathbb{Q}} ( z ) \bigr]\label{eq3.14}
\end{equation}
due to (\ref{eq2.2}), (\ref{eq2.4}), (\ref{eq3.11}), (\ref{eq3.12}) and (\ref{eq3.13}). All in all, combining (\ref{eq3.9})
with (\ref{eq2.4}), (\ref{eq3.10}) and (\ref{eq3.14}), we conclude that
\begin{equation}
M_{t} = \theta^{*} t+ \int_{0}^{t}
\int_{\mathbb{R}_{0}} z d \tilde{N}^{\mathbb{P}} ( s,z )\label{eq3.15}
\end{equation}
which constitutes a Meixner--L\'{e}vy process under $\mathbb{P}$ with
distribution
\begin{equation}
M_{t} \sim\mathcal{ M} \bigl( \alpha, \beta^{*},\delta t,
\mu^{*} t \bigr).\label{eq3.16}
\end{equation}
For this reason, we call the recently introduced measure change
\textit{structure preserving}. Recall that in Section~\ref{sec2} under $\mathbb{Q}$
we observed $M_{t}\sim \mathcal{ M}  ( \alpha,\beta,\delta t,\mu t
 )$ on the other hand. Thus, the proposed measure change does neither
affect the scaling parameter $\alpha$ nor the peakedness parameter $\delta$
whereas both the skewness parameter $\beta$ and the location parameter
$\mu$ are changed. Moreover, the L\'{e}vy triplet of the process $M$
claimed in (\ref{eq3.15}) is given by $ ( \theta^{*},0, \nu^{\mathbb{P}}
 )$. In analogy to the result provided in the sequel of (\ref{eq3.3}), we
remark that under $\mathbb{P}$ it holds
\[
\ln \biggl( \frac{S_{t+ \varDelta}} {S_{t}} \biggr) \cong M_{\varDelta} +b \varDelta\sim
\mathcal{M} \bigl( \alpha, \beta^{*},\delta \varDelta, \bigl(
\mu^{*} +b \bigr) \varDelta \bigr)
\]
due to Lemma~\ref{lem2.1} (a). Here, $\varDelta >0$ is a constant and $b$ is such as
given in (\ref{eq3.3}). Hence, if we specify the Radon--Nikodym function $h  ( z
 )$ like in (\ref{eq3.10}), then the log-returns again are Meixner distributed
under the real-world probability measure $\mathbb{P}$. Further note that
\[
\theta t+ \int_{0}^{t} \int_{\mathbb{R}_{0}}
z d \tilde{N}^{\mathbb{Q}} ( s,z ) = M_{t} = \theta^{*} t+
\int_{0}^{t} \int_{\mathbb{R}_{0}} z d
\tilde{N}^{\mathbb{P}} ( s,z )
\]
holds $\mathbb{P}$- respectively $\mathbb{Q}$-almost surely for all $t\in
 [ 0,T  ]$ thanks to (\ref{eq2.1}) and (\ref{eq3.15}). We obtain the following
expressions for the variance, skewness and kurtosis of $M_{t}$ under~$\mathbb{P}$
\begin{align*}
\mathbh{Var}_{\mathbb{P}} [ M_{t} ] &= \frac{\delta t} {2}
\frac{\alpha^{2}} {\cos^{2}  ( {\beta^{*}} / {2}  )},\qquad \mathbb{S}_{\mathbb{P}} [ M_{t} ] = \sqrt{{2} / {
( \delta t )}} \sin \bigl( {\beta^{*}} / {2} \bigr),\\
\mathbb{K}_{\mathbb{P}} [ M_{t} ] &=3+ \frac{2- \cos  ( \beta^{*}  )} {\delta
t}
\end{align*}
while under $\mathbb{P}$ the density and characteristic function of $M_{t}$
are such as given in (\ref{eq2.6}) and (\ref{eq2.9}) but with $\mu$ and $\beta$ therein
replaced by $\mu^{*}$ and $\beta^{*}$, respectively.

\subsubsection{A generalized structure preserving measure
change}\label{sec3.1.1}

In this section, we present a generalized structure preserving measure
change from the risk-neutral to the physical probability measure. Recall
that the measure change proposed above only affects the skewness parameter
$\beta$ and the location parameter $\mu$ whereas both the scaling parameter
$\alpha$ and the peakedness parameter $\delta$ remain untouched. From a
practical point of view, this fact might be regarded as an advantage, as
there is no need to recalibrate the parameters $\alpha$ and $\delta$ when
changing from the risk-neutral to the physical probability measure.
Conversely, the described feature might likewise cause some difficulties
when it comes to calibrating under $\mathbb{P}$, since the values of the
parameters $\alpha$ and $\delta$ have to be the same as under $\mathbb{Q}$
yielding some loss of flexibility. To avoid this disadvantage, we now
propose a generalized measure change which affects each of the four
parameters of the Meixner distribution. For this purpose, we presently
require that the Meixner--L\'{e}vy measure under $\mathbb{P}$ is of the
form
\begin{equation}
d \nu^{\mathbb{P}} ( z ) = \delta^{*} \frac{e^{{\beta^{*} z} /
{\alpha^{*}}}} {z \sinh  ( {\pi z} / {\alpha^{*}}  )} dz\label{eq3.17}
\end{equation}
[cf. Equation (\ref{eq2.4})] with new parameters $\alpha^{*} >0$, $\beta^{*} \in  (
-\pi,\pi  )$ and $\delta^{*} >0$ which are different from $\alpha$,
$\beta$ and $\delta$ introduced previously under $\mathbb{Q}$. Following
this approach, we are led to the equality
\[
e^{h  ( z  )} = \frac{\delta^{*} \sinh  ( {\pi z} / {\alpha}
 )} {\delta \sinh  ( {\pi z} / {\alpha^{*}}  )} \exp \biggl\{ \biggl( \frac{\beta^{*}} {\alpha^{*}} -
\frac{\beta} {\alpha} \biggr) z \biggr\}
\]
due to (\ref{eq3.17}), (\ref{eq3.8}) and (\ref{eq2.4}). Taking the logarithm in the latter
equation, we receive
\begin{equation}
h ( z ) = \biggl( \frac{\beta^{*}} {\alpha^{*}} - \frac{\beta}{
\alpha} \biggr) z+ \ln \biggl(
\frac{\delta^{*} \sinh  ( {\pi z} /
{\alpha}  )} {\delta \sinh  ( {\pi z} / {\alpha^{*}}  )} \biggr)\label{eq3.18}
\end{equation}
which corresponds to (\ref{eq3.10}) above. Note that (\ref{eq3.18}) is well-defined for all
$z\in \mathbb{R}_{0}$ and that we obtain (\ref{eq3.10}), if we take $\alpha^{*}
=\alpha$ and $\delta^{*} =\delta$ in (\ref{eq3.18}). Hence, (\ref{eq3.10}) is a special
case of (\ref{eq3.18}). Moreover, if we take $\alpha^{*} =\alpha$ and $\beta^{*}
=\beta$ in (\ref{eq3.18}), then we get $h  ( z  ) = \ln \delta^{*} - \ln
\delta$ which is constant and independent of $z$. We summarize our findings
in the following proposition.

\begin{proposition}\label{prop3.1} Consider the measure change from the
risk-neutral to the physical probability measure with Radon--Nikodym density
process $\varLambda$ such as defined at the beginning of Section~\ref{sec3.1}. Then the new L\'{e}vy measure $\nu^{\mathbb{P}}$
under $\mathbb{P}$ is of Meixner-type again, if and only
if the Radon--Nikodym function $h  ( z  )$ is of the form
(\ref{eq3.18}).
\end{proposition}

In the sequel, we investigate the distributional properties of the
corresponding Meixner--L\'{e}vy process under $\mathbb{P}$ related to the
Radon--Nikodym function $h  ( z  )$ given in (\ref{eq3.18}). A substitution
of (\ref{eq2.2}), (\ref{eq2.4}) and (\ref{eq3.18}) into (\ref{eq3.9}) yields the following L\'{e}vy--It\^{o}
decomposition under $\mathbb{P}$
\begin{equation}
M_{t} = \overline{\theta} t+ \int_{0}^{t}
\int_{\mathbb{R}_{0}} z d \tilde{N}^{\mathbb{P}} ( s,z )\label{eq3.19}
\end{equation}
with deterministic and real-valued drift parameter
\[
\overline{\theta} \coloneqq\mu+\delta\alpha \tan \biggl( \frac{\beta} {2}
\biggr) + \delta^{*} \int_{\mathbb{R}_{0}} \frac{e^{{\beta^{*} z} /
{\alpha^{*}}}} {\sinh  ( {\pi z} / {\alpha^{*}}  )}
dz -\delta \int_{\mathbb{R}_{0}} \frac{e^{{\beta z} / {\alpha}}} {\sinh  ( {\pi z}
/ {\alpha}  )} dz.
\]
The Meixner--L\'{e}vy process $M_{t}$ given in (\ref{eq3.19}) possesses the L\'{e}vy
triplet $ ( \overline{\theta},0, \nu^{\mathbb{P}}  )$ where
$\nu^{\mathbb{P}}$ is such as claimed in (\ref{eq3.17}). In the next step, we
require that $\overline{\theta}$ is of the form (\ref{eq2.2}), i.e.
\[
\overline{\theta} = \overline{\mu} + \delta^{*} \alpha^{*}
\tan \bigl( {\beta^{*}} / {2} \bigr)
\]
with some new location parameter $\overline{\mu} \in\mathbb{R}$. Following
this onset, we deduce
\[
\overline{\mu} \,{=}\,\mu+\delta\alpha \tan \biggl( \frac{\beta} {2} \biggr) -
\delta^{*} \alpha^{*} \tan \biggl( \frac{\beta^{*}} {2}
\biggr) + \delta^{*} \int_{\mathbb{R}_{0}} \frac{e^{{\beta^{*} z} / {\alpha^{*}}}} {\sinh  (
{\pi z} / {\alpha^{*}}  )}
dz -\delta \int_{\mathbb{R}_{0}} \frac{e^{{\beta z} / {\alpha}}} {\sinh  ( {\pi z} / {\alpha}  )} dz.
\]
Hence, if the measure change from $\mathbb{Q}$ to $\mathbb{P}$ is performed
with the Radon--Nikodym function $h  ( z  )$ defined in (\ref{eq3.18}),
then the corresponding Meixner--L\'{e}vy process $M_{t}$ again is Meixner
distributed under $\mathbb{P}$ with parameters
\begin{equation}
M_{t} \sim\mathcal{ M} \bigl( \alpha^{*},
\beta^{*}, \delta^{*} t, \overline{\mu} t \bigr).\label{eq3.20}
\end{equation}
If we compare (\ref{eq3.20}) with (\ref{eq3.16}), we see that in the generalized measure
change related to (\ref{eq3.18}) each of the four parameters of the
Meixner--L\'{e}vy process $M$ is affected.

\section{Cliquet option pricing in a geometric Meixner model}\label{sec4}

This section is devoted to the pricing of cliquet options in the Meixner
stock price model presented in Chapter~\ref{sec3}. Since the Meixner process $Y$ in
(\ref{eq3.5}) above just is a special case of the more general L\'{e}vy process $X$
defined in Eq. (2.3) in \cite{8}, the cliquet option pricing results derived in
\cite{8} simultaneously apply to our present Meixner--L\'{e}vy modeling case. The
details are worked out in the remainder of the current section. Parallel to
\cite{8} and Eq. (1.1) in \cite{3}, we consider a monthly sum cap style cliquet
option with payoff
\[
H_{T} =K+K \max \Biggl\{ g, \sum_{k=1}^{n}
\min \{ c, R_{k} \} \Biggr\}
\]
where $T$ is the maturity time, $K$ denotes the notional (i.e. the initial
investment), $g$ is the guaranteed rate at maturity, $c\geq0$ is the local
cap and $R_{k}$ is the return process given in (\ref{eq3.4}). Recall that the
payoff $H_{T}$ is globally-floored by the constant $K  ( 1+g  )$
and locally-capped by $c$. By a case distinction, we get
\[
H_{T} =K \max \Biggl\{ 1+g,1+ \sum_{k=1}^{n}
\min \{ c, R_{k} \} \Biggr\} =K \Biggl( 1+g+ \max \Biggl\{ 0, \sum
_{k=1}^{n} Z_{k} \Biggr\} \Biggr)
\]
where for all $k\in  \{ 1,\dots,n  \}$ the appearing objects
\begin{equation}
Z_{k} \coloneqq \min \{ c, R_{k} \} - {g} / {n}\label{eq4.1}
\end{equation}
are independent and identically distributed (i.i.d.) random variables. Note
that $R_{k}$ is $\mathcal{F}_{t_{k}}$-measurable such that $H_{T}$ is
$\mathcal{F}_{t_{n}}$-measurable. Since $t_{n} \leq T$, it holds
$\mathcal{F}_{t_{n}} \subseteq \mathcal{F}_{T}$ such that $H_{T}$
constitutes an $\mathcal{F}_{T}$-measurable claim. As before, let us denote
the constant interest rate by $r>0$. Then the price at time $t\leq T$ of a
cliquet option with payoff $H_{T}$ at maturity $T$ is given by the
discounted risk-neutral conditional expectation of the payoff, i.e.
\[
C_{t} = e^{-r  ( T-t  )} \mathbb{E}_{\mathbb{Q}} ( H_{T}
\mid \mathcal{F}_{t} ).
\]
Combining the latter equations, we obtain
\begin{equation}
C_{0} =K e^{-rT} \Biggl( 1+g+ \mathbb{E}_{\mathbb{Q}}
\Biggl[ \max \Biggl\{ 0, \sum_{k=1}^{n}
Z_{k} \Biggr\} \Biggr] \Biggr)\label{eq4.2}
\end{equation}
which shows that the considered cliquet option with payoff $H_{T}$
essentially is a plain-vanilla call option with strike zero written on the
basket-style underlying $\sum_{k=1}^{n} Z_{k}$.

\begin{proposition}[Cliquet option price]\label{prop4.1}
Let $k\in  \{1,\dots,n  \}$ and consider the independent and identically
distributed random variables $Z_{k} = \min  \{ c, R_{k}  \} - {g}
/ {n}$ where $c\geq0$ is the local cap, $R_{k}$
is the return process given in (\ref{eq3.4}) and $g$
is the guaranteed rate at maturity. Denote the maturity time by
$T$, the notional by $K$ and the risk-less interest rate
by $r$. Then the price at time zero of a cliquet option with
payoff $H_{T}$ can be represented as
\begin{equation}
C_{0} =K e^{-rT} \Biggl( 1+g+ \frac{n} {2}
\mathbb{E}_{\mathbb{Q}} [ Z_{1} ] + \frac{1} {\pi} \int
_{0^{ +}}^{\infty} \frac{1- Re  (
\phi_{Z}  ( x  )  )} {x^{2}} dx \Biggr)\label{eq4.3}
\end{equation}
where $Re$ denotes the real part and the characteristic
function $\phi_{Z}  ( x  )$ is defined via
\begin{equation}
\phi_{Z} ( x ) \coloneqq \prod_{k=1}^{n}
\phi_{Z_{k}} ( x ) = \prod_{k=1}^{n}
\mathbb{E}_{\mathbb{Q}} \bigl[ e^{ix Z_{k}} \bigr] = \bigl(
\phi_{Z_{1}} ( x ) \bigr)^{n} = \bigl( \mathbb{E}_{\mathbb{Q}}
\bigl[ e^{ix Z_{1}} \bigr] \bigr)^{n}.\label{eq4.4}
\end{equation}
\end{proposition}

\begin{proof} See the proof of Prop. 3.1 in \cite{3}, respectively of Prop.
3.1 in \cite{8}. \end{proof}

In the subsequent sections, we derive explicit expressions for $\phi_{Z}
 ( x  )$ and $\mathbb{E}_{\mathbb{Q}}  [ Z_{1}  ]$
appearing in the pricing formula (\ref{eq4.3}). As before, we stick to the
presumption of equidistant resetting times and set $\tau= t_{k} - t_{k-1}$
for all $k\in  \{ 1,\dots,n  \}$ in the following.

\subsection{Cliquet option pricing with distribution
functions}\label{sec4.1}

Let us first apply a method involving probability distribution functions
(cf. \cite{3} and Section~3.1 in \cite{8}). We initially investigate the treatment of
$\phi_{Z}  ( x  )$ defined in (\ref{eq4.4}).

\begin{proposition}\label{prop4.2} Let $Y_{\tau} \sim\mathcal{ M}  (
\alpha,\beta,\delta\tau,  ( \mu+b  ) \tau  )$ and
suppose that $Z_{k} =\break \min  \{ c, R_{k}  \} - {g} / {n}$
where $k\in  \{ 1,\dots,n  \}$. Then the
characteristic function of $Z_{k}$ under $\mathbb{Q}$ can
be represented as
\begin{align}
\phi_{Z_{k}} ( x ) &= e^{-ix  ( 1+ {g} / {n}  )} \Biggl( e^{ix  ( 1+c  )} \int
_{1+c}^{\infty} f_{Y_{\tau}} ( u ) du \notag\\
&\quad+ \int
_{-\infty}^{0} f_{Y_{\tau}} ( u ) du + \int
_{0}^{1+c} e^{ixu} f_{Y_{\tau}} ( u ) du
\Biggr)\label{eq4.5}
\end{align}
where
\begin{equation}
f_{Y_{\tau}} ( u ) = \frac{ ( 2 \cos  ( {\beta} / {2}
 )  )^{2\delta\tau}} {2\pi \alpha \varGamma  ( 2\delta\tau
 )} e^{{\beta  ( u-  ( \mu+b  ) \tau  )} / {\alpha}} \biggl\llvert
\varGamma \biggl( \delta\tau+i \frac{u-  ( \mu+b  ) \tau}{
\alpha} \biggr) \biggr\rrvert
^{2}\label{eq4.6}
\end{equation}
constitutes the probability density function of the
Meixner--L\'{e}vy process $Y$ given in (\ref{eq3.5}) and $b$
is the real-valued constant claimed in (\ref{eq3.3}).
\end{proposition}

\begin{proof} By similar arguments as in the proof of Prop. 3.2 in \cite{8},
we obtain
\begin{equation}
\phi_{Z_{k}} ( x ) = e^{-ix  ( 1+ {g} / {n}  )} \Biggl( e^{ix  ( 1+c  )} -ix \int
_{0}^{1+c} e^{ixw} \mathbb{ Q} (
R_{k} \leq w-1 ) dw \Biggr).\label{eq4.7}
\end{equation}
Using (\ref{eq3.6}) and the definition of the distribution function, we get for the
last integral in (\ref{eq4.7})
\[
\int_{0}^{1+c} e^{ixw} \mathbb{ Q} (
R_{k} \leq w-1 ) dw = \int_{0}^{1+c} \int
_{-\infty}^{\ln  ( w  )} e^{ixw} f_{Y_{\tau}} ( u ) du
dw
\]
where $Y_{\tau} \sim\mathcal{ M}  ( \alpha,\beta,\delta\tau,  (
\mu+b  ) \tau  )$ is the Meixner--L\'{e}vy process given in (\ref{eq3.5})
and $f_{Y_{\tau}}  ( u  )$ constitutes the probability density
function of $Y_{\tau}$ under $\mathbb{Q}$ claimed in (\ref{eq4.6}). Applying
Fubini's theorem and hereafter splitting up the resulting outer integral,
we deduce
\begin{align*}
&\int_{0}^{1+c} e^{ixw} \mathbb{ Q} (
R_{k} \leq w-1 ) dw \\
&\quad= \int_{-\infty}^{0} \int
_{0}^{1+c} e^{ixw} f_{Y_{\tau}} ( u ) dw
du + \int_{0}^{1+c} \int_{u}^{1+c}
e^{ixw} f_{Y_{\tau}} ( u ) dw du.
\end{align*}
We next compute the emerging $dw$-integrals and finally substitute the
resulting expression into (\ref{eq4.7}) which yields (\ref{eq4.5}). \end{proof}

If we insert (\ref{eq4.5}) into (\ref{eq4.4}), we receive a representation for the
characteristic function $\phi_{Z}  ( x  )$. Let us proceed with
the computation of $\mathbb{E}_{\mathbb{Q}}  [ Z_{k}  ]$.

\begin{proposition}\label{prop4.3} Suppose that $Z_{k} = \min  \{ c,
R_{k}  \} - {g} / {n}$ where $k\in  \{ 1,\dots,n
 \}$. Then the first moment of $Z_{k}$ under
$\mathbb{Q}$ is given by
\begin{equation}
\mathbb{E}_{\mathbb{Q}} [ Z_{k} ] =-1- \frac{g} {n} + ( 1+c )
\int_{1+c}^{\infty} f_{Y_{\tau}} ( u ) du + \int
_{0}^{1+c} u f_{Y_{\tau}} ( u ) du\label{eq4.8}
\end{equation}
where $f_{Y_{\tau}}  ( u  )$ is the probability
density function of $Y_{\tau}$ under $\mathbb{Q}$ given
in (\ref{eq4.6}).
\end{proposition}

\begin{proof} In accordance to Prop. 2.4 in \cite{4}, we have
\begin{equation}
\mathbb{E}_{\mathbb{Q}} [ Z_{k} ] = \frac{1} {i}
\frac{\partial} {\partial x} \bigl( \phi_{Z_{k}} ( x ) \bigr) \big\rrvert
_{x=0}.\label{eq4.9}
\end{equation}
A substitution of (\ref{eq4.5}) into (\ref{eq4.9}) instantly yields (\ref{eq4.8}). \end{proof}

As mentioned in Section~\ref{sec2}, we recall that the cumulative distribution
function (cdf) of the Meixner-L\'{e}vy process $Y_{t}$, i.e.
\[
F_{Y_{t}} ( x ) \coloneqq \int_{-\infty}^{x}
f_{Y_{t}} ( u ) du
\]
does not possess a closed form representation, but it can be computed
efficiently with numerical methods. Also note that all integrals appearing
in (\ref{eq4.5}) and (\ref{eq4.8}) are finite, since $f_{Y_{\tau}}  (
\boldsymbol{\cdot}  )$ constitutes a probability density function
while the Meixner--L\'{e}vy process $Y_{\tau}$ possesses moments of all
orders (cf. \cite{13}).

\subsection{Cliquet option pricing with Fourier transform
techniques}\label{sec4.2}

There is an alternative method to derive expressions for
$\mathbb{E}_{\mathbb{Q}}  [ Z_{k}  ]$, $\phi_{Z}  ( x  )$
and $C_{0}$ involving Fourier transforms and the L\'{e}vy--Khinchin formula.
In the following, we present this method which has firstly been proposed in
\cite{8} in a cliquet option pricing context.

\begin{proposition}\label{prop4.4} Let $Y_{t} \sim\mathcal{ M}  (
\alpha,\beta,\delta t,  ( \mu+b  ) t  )$ be the
Meixner--L\'{e}vy process considered in (\ref{eq3.5}). Suppose that
$Z_{k} = \min  \{ c, R_{k}  \} - {g} / {n}$ where $k\in
 \{ 1,\dots,n  \}$ and let $\vartheta>0$ be a
finite real-valued dampening parameter. Then the first moment of $Z_{k}$
under $\mathbb{Q}$ can be represented as
\begin{equation}
\mathbb{E}_{\mathbb{Q}} [ Z_{k} ] =c- \frac{g} {n} -
\frac{1}{
2\pi} \int_{\mathbb{R}} \frac{ ( c+1  )^{1+\vartheta+iy}} { (
\vartheta+iy  )  ( 1+\vartheta+iy  )}
\phi_{Y_{\tau}} ( i\vartheta-y ) dy\label{eq4.10}
\end{equation}
where the characteristic function $\phi_{Y_{\tau}}$ is
given by
\begin{equation}
\phi_{Y_{\tau}} ( i\vartheta-y ) = e^{-  ( \vartheta+iy
 )  ( \mu+b  ) \tau} \biggl(
\frac{\cos  ( {\beta} / {2}
 )} {\cosh  ( { ( i  ( \alpha\vartheta-\beta  )
-\alpha y  )} / {2}  )} \biggr)^{2\delta\tau}.\label{eq4.11}
\end{equation}
\end{proposition}

\begin{proof} The proof follows the same lines as the proof of Prop. 3.4
in \cite{8}. From (\ref{eq4.1}) and the equality
\[
\min \{ c, R_{k} \} =c- [ c- R_{k} ]^{+}
\]
we deduce
\[
\mathbb{E}_{\mathbb{Q}} [ Z_{k} ] =c- {g} / {n} -
\mathbb{E}_{\mathbb{Q}} \bigl[ ( c- R_{k} )^{+} \bigr].
\]
Taking (\ref{eq3.4}) into account, we next receive
\[
\mathbb{E}_{\mathbb{Q}} [ Z_{k} ] =c- {g} / {n} -
\mathbb{E}_{\mathbb{Q}} \bigl[ \bigl( c+1- e^{Y_{\tau}} \bigr)^{+}
\bigr]
\]
where $\tau= t_{k} - t_{k-1}$ and $Y$ is the real-valued Meixner--L\'{e}vy
process given in (\ref{eq3.5}). With a finite and real-valued dampening parameter
$\vartheta>0$ we define the function
\[
\varphi ( u ) \coloneqq e^{\vartheta u} \bigl( c+1- e^{u}
\bigr)^{+}.
\]
Since $\varphi\in \mathcal{L}^{1}  ( \mathbb{R}  )$, its Fourier
transform exists and reads as
\[
\hat{\varphi} ( y ) = \frac{ ( c+1  )^{1+\vartheta+iy}}{
 ( \vartheta+iy  )  ( 1+\vartheta+iy  )}.
\]
Using the inverse Fourier transform along with Fubini's theorem, we get
\[
\mathbb{E}_{\mathbb{Q}} \bigl[ \bigl( c+1- e^{Y_{\tau}}
\bigr)^{+} \bigr] = \mathbb{E}_{\mathbb{Q}} \bigl[ e^{-\vartheta Y_{\tau}}
\varphi ( Y_{\tau} ) \bigr] = \frac{1} {2\pi} \int_{\mathbb{R}}
\hat{\varphi} ( y ) \mathbb{E}_{\mathbb{Q}} \bigl[ e^{-  ( \vartheta+iy
 ) Y_{\tau}} \bigr] dy
\]
which implies (\ref{eq4.10}). The expression for the characteristic function
$\phi_{Y_{\tau}}$ given in (\ref{eq4.11}) can directly be obtained by virtue of
(\ref{eq2.9}). \end{proof}

Our argumentation in the proof of Proposition~\ref{prop4.4} motivates the following
considerations.

\begin{proposition}\label{prop4.5} Let $Y_{t} \sim\mathcal{ M}  (
\alpha,\beta,\delta t,  ( \mu+b  ) t  )$ be the
Meixner--L\'{e}vy process presented in (\ref{eq3.5}). Suppose that
$Z_{k} = \min  \{ c, R_{k}  \} - {g} / {n}$ with $k\in
 \{ 1,\dots,n  \}$ and $c\geq0$. Then the
characteristic function of $Z_{k}$ under $\mathbb{Q}$
reads as
\begin{equation}
\phi_{Z_{k}} ( x ) = e^{-ix {g} / {n}} \Biggl( e^{ixc} + \int
_{-\infty}^{\ln  ( 1+c  )} \bigl[ e^{ix  ( e^{u} -1
 )} - e^{ixc}
\bigr] f_{Y_{\tau}} ( u ) du \Biggr)\label{eq4.12}
\end{equation}
where the probability density function $f_{Y_{\tau}}$ of $Y_{\tau}$
under $\mathbb{Q}$ is such as given in (\ref{eq4.6}).
\end{proposition}

\begin{proof} Similar computations as in the proof of Prop. 3.5 in \cite{8}
yield (\ref{eq4.12}). \end{proof}

There is an alternative method involving (\ref{eq4.9}) to derive an expression for
$\mathbb{E}_{\mathbb{Q}}  [ Z_{k}  ]$ which is presented in the
following.

\begin{corollary}\label{cor4.1} In the setup of Proposition~\ref{prop4.5},
we receive the representation
\begin{equation}
\mathbb{E}_{\mathbb{Q}} [ Z_{k} ] =c- \frac{g} {n} + \int
_{-\infty}^{\ln  ( 1+c  )} \bigl[ e^{u} -1-c \bigr]
f_{Y_{\tau}} ( u ) du.\label{eq4.13}
\end{equation}
\end{corollary}

\begin{proof} The claimed representation immediately follows from Eq.
(3.16) in \cite{8}. \end{proof}

Inspired by the Fourier transform techniques applied in the proof of
Proposition~\ref{prop4.4}, we now focus on the derivation of an alternative
representation for the cliquet option price $C_{0}$ given in (\ref{eq4.2}). The
corresponding result reads as follows.

\begin{theorem}[Fourier transform cliquet option price]\label{thm4.1} Let
$k\in  \{ 1,\dots,n  \}$ and consider the independent and
identically distributed random variables $Z_{k} \,{=}\, \min  \{ c, R_{k}
 \} - {g} / {n}$ where $c\geq0$ is the local cap,
$g$ is the guaranteed rate at maturity and $R_{k}$ is the
return process defined in (\ref{eq3.4}). For $n \in\mathbb{N}$
we set $\varrho\coloneqq nc-g$ and denote the maturity
time by $T$, the notional by $K$ and the riskless
interest rate by $r$. Let $Y_{t} \sim\mathcal{ M}  (
\alpha,\beta,\delta t,  ( \mu+b  ) t  )$ be the
Meixner--L\'{e}vy process given in (\ref{eq3.5}). Then the price at
time zero of a cliquet option paying
\[
H_{T} =K \Biggl( 1+g+ \max \Biggl\{ 0, \sum
_{k=1}^{n} Z_{k} \Biggr\} \Biggr)
\]
at maturity can be represented as
\begin{align}
C_{0} &=K e^{-rT} \Biggl[ 1+g+ \int_{0^{ +}}^{\infty}
\frac{1+iy\varrho-
e^{iy\varrho}} {2\pi y^{2}}\notag \\
&\quad\times\Biggl( 1+ \int_{-\infty}^{\ln  ( 1+c
 )} \bigl[
e^{iy  ( e^{u} -1-c  )} -1 \bigr] f_{Y_{\tau}} ( u ) du \Biggr)^{n} dy
\Biggr]\label{eq4.14}
\end{align}
where $f_{Y_{\tau}}  ( u  )$ constitutes the
probability density function claimed in (\ref{eq4.6}).
\end{theorem}

\begin{proof} The proof of Theorem 3.7 in \cite{8} here applies equally, if we
replace $f_{X_{\tau}}$ therein by $f_{Y_{\tau}}$. \end{proof}

\section{Conclusions}\label{sec5}

In this paper, we investigated the pricing of a monthly sum cap style
cliquet option with underlying stock price modeled by a geometric pure-jump
Meixner--L\'{e}vy process. In Section~\ref{sec2}, we compiled various facts on the
Meixner distribution and the related class of stochastic Meixner--L\'{e}vy
processes. In Section~\ref{sec3}, we introduced a stock price model driven by a
Meixner--L\'{e}vy process and established a customized structure preserving
measure change from the risk-neutral to the physical probability measure.
Moreover, we obtained semi-analytic expressions for the cliquet option
price by using the probability distribution function of the driving
Meixner--L\'{e}vy process in Section~\ref{sec4.1} and by an application of Fourier
transform techniques in Section~\ref{sec4.2}. To read more on cliquet option pricing
in a jump-diffusion L\'{e}vy model, the reader is referred to the
accompanying article \cite{8}.





\end{document}